%% file: main.tex
\documentclass[runningheads]{llncs}

\usepackage{enumitem}
\usepackage{hyperref}
\usepackage{algorithm}
\usepackage{algorithmic}
\usepackage{xspace}

\usepackage{amsmath,amssymb,dsfont}
\usepackage{tikz}

\usepackage{listings}
\lstdefinelanguage{coq}
{
  morekeywords ={Definition, Lemma, Theorem, forall, exists, Inductive,
    CoInductive, Type, Class, Hypothesis, Fixpoint, Record, if, then, else},
  sensitive=true,
  morecomment =[s]{(*}{*)},
  escapeinside={(@}{@)},
  emph={Prop}, emphstyle=\bf,
}

\lstdefinestyle{coqstyle}{
  language=coq, 
  commentstyle=\sl, 
  keywordstyle=\bf, mathescape=true,
  basicstyle=\footnotesize\tt
}
\input{commands}

\begin{document}

\title{Revisited Convergence of Dolev \emph{et al}’s BFS Spanning Tree
  Algorithm}

\author{Karine Altisen\orcidID{0000-0001-8344-1853} \and
  Marius Bozga\orcidID{0000-0003-4412-5684}}

\authorrunning{K. Altisen, M. Bozga}

\institute{Univ. Grenoble Alpes, CNRS, Grenoble
  INP\footnote{Institute of Engineering Univ. Grenoble Alpes},
  VERIMAG, 38000 Grenoble, France
  \email{\{Karine.Altisen,Marius.Bozga\}@univ-grenoble-alpes.fr}\\
  \url{http://www-verimag.imag.fr/}}

\maketitle

\begin{abstract}
  \input{abstract}
\end{abstract}

\section{Introduction}
\label{sec:introduction}
\input{intro}

\section{Dolev \emph{et al}'s BFS Spanning Tree Algorithm}
\label{sec:dolev}
\input{dolev}

\section{The PADEC Framework}
\label{sec:padec}
\input{padec}

\section{Overview of the Proof}
\label{sec:proof}
\input{proof}

\section{A Decreasing Potential Function for $d$-steps}
\label{sec:dstep-potential}
\input{dstep-potential}

\section{Conclusion}
\label{sec:conclusion}
\input{conclusion}

\bibliographystyle{splncs04}
\bibliography{biblio}

\end{document}

%% file: commands.tex
\newcommand{\BFS}{\ensuremath{\mathcal{B\!F\!S}}\xspace}
\newcommand{\Root}{\ensuremath{\boldsymbol{r}}\xspace}

\newcommand{\ie}{\emph{i.e.}}
\newcommand{\eg}{\emph{e.g.}}
\newcommand{\nb}{\emph{n.b.}}

\newcommand{\isdef}{\stackrel{def}=}

\newcommand{\dist}[2]{\mathsf{dist}(#1,#2)}

\newcommand{\maxd}{\max\nolimits_d}
\newcommand{\mind}{\min\nolimits_d}
\newcommand{\sumd}{\mathop{\mathrm{sum}_d}}

\newcommand{\dleq}{ \le_d }
\newcommand{\dtop}[1]{#1^{\top}}
\newcommand{\dbot}[1]{#1^{\bot}}

\newcommand{\crank}[2]{\mathit{rank}_{#1}(#2)}
\newcommand{\csmooth}[2]{\mathit{smooth}_{#1}(#2)}
\newcommand{\ssmooth}[1]{\mathit{smooth}(#1)}

\newcommand{\nsset}[2]{E_{#1,#2}}

\newcommand{\dpot}[1]{\pi_{#1}}

\newcommand{\dstep}{\rightarrow_d}
\newcommand{\dstepstar}{\rightarrow_d^*}

\newcommand{\Step}{\ensuremath{\mathsf{Step}}\xspace}
\newcommand{\RStep}{\Step^{(r)}}
\newcommand{\DStep}{\Step^{(d)}}
\newcommand{\PStep}{\Step^{(p)}}

\newcommand{\Acc}{\mathsf{Acc}}

\newcommand{\WellFounded}{\mathsf{WellFounded}}

\newcommand{\xstep}[1]{\xrightarrow{#1}}

\newcommand{\Nodes}{\ensuremath{\mathit{Nodes}}\xspace}
\newcommand{\Channels}{\ensuremath{\mathit{Channels}}\xspace}
\newcommand{\States}{\ensuremath{\mathit{States}}\xspace}
\newcommand{\nodes}{\Nodes}
\newcommand{\channels}{\Channels}
\newcommand{\Env}{\ensuremath{\Gamma}\xspace}
\newcommand{\Edges}{\ensuremath{\mathit{Edges}}\xspace}
\newcommand{\edges}{\Edges}
\renewcommand{\r}{\Root}

%% file: abstract.tex
We provide a constructive proof for the convergence of Dolev \emph{et
al}'s BFS spanning tree algorithm running under the general assumption
of an unfair daemon.  Already known proofs of this algorithm are
either using non-constructive principles (e.g., proofs by
contradiction) or are restricted to less general execution daemons
(e.g., weakly fair).  In this work, we address these limitations by
defining the well-founded orders and potential functions ensuring
convergence in the general case.  The proof has been fully formalized
in PADEC, a Coq-based framework for certification of
self-stabilization algorithm.

%% file: intro.tex
To obtain properties about distributed systems is a difficult task.
Indeed, such systems generally involve an arbitrary number of
participants, interconnected and communicating according to arbitrary
or specific topologies.  Their execution could be suject to various
assumptions about the degree of asynchronism between participants.
The context in which those systems are considered is increasingly
complex, for example, large scale distributed systems, made of
hererogenous devices, running in higly dynamic networks. Above all,
every nowadays distributed system requires some fault tolerance
properties which are particularly difficult to handle.  Therefore, the
proofs for distributed algorithms quickly become complex and may lead
to errors \cite{Lamport2012}.

In this context, the common practice for establishing the correctness
of distributed algorithms was to provide proofs on paper. But,
computer-aided validation tools are being developped as an answer to
improve those practices.  In particular, many approaches are based on
model-checking (see \eg, \cite{DBLP:journals/sttt/BertrandKLW21},
\cite{Tsuchiya01}) or controller synthesis (see \eg,
\cite{Ebnenasir22}, \cite{DBLP:journals/dc/VolkBKA22}).  However,
those methods usually face the huge execution space, restricting
their applicability. Furthermore, they either require to completely
define the full context of execution (such as the number of
participants, the topology of the network, etc) or to provide
parametric solutions which are often restricted due to undecidability
limitations \cite{DBLP:series/synthesis/2015Bloem}.

On the other hand, the use of a proof assistant such as Coq \cite{coq}
allows to formalize the proofs and automatically check (a.k.a certify)
their correctness for a given distributed algorithm, for every value
of parameters such as the topology or the number of participants.
Note that a proof assistant is not meant to derive any proof: the
proof designer comes with (at least) a sketch of proof and the tool is
used as a guidance to enumerate cases, avoid flaws and clarify the
assumptions. Once completed, the proof has been fully checked by the
tool and its soundness is guaranteed. Many frameworks have been
developped to certify the proof of some distributed algorithms using
various proof assistant such as Coq \cite{ACD7,Courtieu02,pactole},
TLA+ \cite{CDLMRV12c,tla} or Isabel/HOL \cite{berni,JM05tr,KNR12}.

We focus here on distributed self-stabilizing
algorithms. Self-stabilization \cite{D74j} is a lightweight fault
tolerance property to withstand transient faults: once such faults hit
a self-stabilizing system, this one is guaranteed to recover a correct
behavior within finite time. Note that no assumption is made on the
nature of the transient faults (memory corruption, topology changes,
etc). But, once those faults cease, there is a finite period - the
stabilization time - during which the system may misbehave (notably
its safety guaranties are no longer ensured during this recevory
period).  In this paper, we are interested in the proof of
self-stabilization of a BFS spanning tree algorithm by Dolev \emph{et
al} \cite{DIM93}.

\paragraph{Related Work}

The correctness of several non fault-tolerant distributed algorithms
have been certified, (see \eg, \cite{CF11j,HESSELINK20131622}).
Certification of fault-tolerant, yet non self-stabilizing, distributed
systems has been addressed using various proof assistants, {\em e.g},
in Isabel/HOL \cite{CM09j,JM05tr,berni,KNR12}, TLA+
\cite{CDLMRV12c,DFGL13c}, Coq \cite{RVVV18c}, NuPRL~\cite{RGBC17j}.
This approach is called {\em robust} fault tolerance; it masks the
effect of the faults (whereas self-stabilization is non-masking by
essence). In the robust approach, many results are related to
agreement problems, such as consensus or state-machine replication, in
fully connected networks; and many works only certify the safety
property of the considered
problem (see \eg, \cite{CDLMRV12c,DFGL13c,RGBC17j,RVVV18c}).
However, both liveness and safety properties are certified
in \cite{CM09j,KNR12,berni}.
Finally, robust fault tolerance has been also considered in the
context of mobile robot computing: using the PACTOLE Coq framework,
impossibility results for swarms of robots that are subjected to
Byzantine faults have been certified \cite{bouzid13sss,CRTU15}.

Several frameworks to certify self-stabilizing algorithms using the
Coq proof assistant have been proposed, \eg, \cite{Courtieu02,ACD7}.
In particular, the PADEC Coq library provides a framework to develop
proofs of self-stabilizing algorithms written in the atomic state
model \cite{D74j}, and allows many various assumptions defined in
the litterature. For instance, the asynchronism of the system can be
defined using several levels of fairness.
Notably, it includes support and use cases that prove the composition
of self-stabilizing algorithms \cite{DBLP:conf/forte/AltisenCD19},
their time complexity in steps \cite{ACD21c} and
rounds \cite{AltisenCD23} (Rounds provide a measure of time taking
into account the parallelism of the system whereas steps provide a
sequential measure).

The research on proving termination (\ie\ convergence) of distributed
algorithms is extremely vast. Usually, formal techniques such as
ranking/potential functions (\eg, \cite{ACD21c}), well-founded orders
(\eg, \cite{ACD7}) provide direct means to prove the termination
over all executions.  Alternatively, by the principle of the excluded
middle, termination follows from a proof of the absence of
non-terminating runs.  That is, proofs of termination by contradiction
usually exhibit contradictions in the case non-terminating executions
are presumed possible.  The principle of the excluded middle is,
however, not allowed in constructive proofs \eg, based on
intuitionistic logics and hence impossible to use in constructive
proof assistants such as Coq \cite{coqart}
(unless changing the basic set of axioms).

\paragraph{Contributions}
In this paper, we revisit the proof of convergence of the
self-stabilizing Dolev \emph{et al} BFS spanning tree
algorithm \cite{DIM93}.  The first proof of convergence of this
algorithm has been provided in \cite{DIM93}, that is, the same paper
where the algorithm has been introduced; but this proof is restricted
to some restricted fairness assumptions.  Since then, other proofs
have been proposed.
As \cite{DIM93}, \cite{thebook} proves the self-stabilization of the
algorithm under the mild assumption of a weakly fair daemon, that is,
by restricting the asynchronism along the executions of the algorithm;
this proof has been formalized, developped and mechanically checked in
PADEC \cite{AltisenCD23}.
While relaxing the fairness assumption,
\cite{thebook} also provides a non-constructive proof, by
contradiction, working under the explicit assumption that the diameter
of the graph is a priori known and used as a bound on some of the
variables of the algorithm.

In this paper, still with no fairness assumption, we provide, a
contrario, a constructive proof of the result which has been fully
formalized and mechanically checked in PADEC.
Our contribution is twofold:
\begin{itemize}
\item We provide the first constructive proof of the convergence of
  the Dolev \emph{et al} BFS Spanning Tree algorithm under the most
  general execution assumptions (\ie, unfair daemon, unbounded
  variables). The proof exploits a novel potential function, allowing
  for a finer comprehension of the system executions towards
  convergence.
\item The convergence proof has been fully formalized and
  automatically checked using the PADEC framework.  The result can be
  therefore fully trusted and moreover, illustrates the capabilities
  of the PADEC framework to formally handle distributed algorithms and
  their properties.
\end{itemize}

\paragraph{Coq Development.}  The development for this contribution
represents about 5,155 lines of Coq code (\lstinline|#loc|, as
measured by \lstinline|coqwc|), precisely
\lstinline|#loc: spec = 1,111;|
\lstinline|proof = 3,662;| \lstinline|comments = 382|.  It
 is available as an online browsing documentation at
\url{http://www-verimag.imag.fr/~altisen/PADEC}. We encourage the
reader to visit this web-page for a deeper understanding of our work.

\paragraph{Organization.} The paper is organized as follows.
Section \ref{sec:dolev} recalls the Dolev \emph{et al} algorithm for
the construction of BFS spanning trees.  Section \ref{sec:padec}
provides a brief overview of the PADEC framework and the formal
encoding of the above-mentioned algorithm and its relevant properties.
Section \ref{sec:proof} provides the proof of convergence.  Section
\ref{sec:dstep-potential} elaborates on the definition of the
potential function over system configurations, that is, the key
ingredient ensuring that the proof is constructive and therefore
representable in PADEC.  Finally, Section \ref{sec:conclusion}
concludes and provides directions for future work.

%% file: dolev.tex
The Dolev \emph{et al}'s BFS algorithm~\cite{DIM93} is a
self-stabilizing distributed algorithm that computes a BFS spanning
tree in an arbitrary rooted, connected, and bidirectional network.  By
``bidirectional'', we mean that each node can both transmit and
acquire information from its adjacent nodes in the network topology,
\ie, its neighbors.  The algorithm being distributed, these are the
only possible direct communications.  ``Rooted'' indicates that a
particular node, called the root and denoted by \Root, is
distinguished in the network. As in the present case, algorithms for
rooted networks are usually semi-anonymous: all nodes have the same
code except the root.

This algorithm was initially written in the Read/Write atomicity
model. We study, here, a straightforward translation into the
\emph{atomic-state model}, denoted hereafter by \BFS, and presented as
Algorithm~\ref{alg}. Notice that, as in the original presentation
\cite{DIM93} and contrarily to other adapations (see \eg,
\cite{thebook}), the variables are not assumed to be bounded.

\begin{algorithm}[htp]
  
  \textbf{Constant Local Inputs:} \hfill\
  
  \begin{tabular}{l}
    $p.\mathit{neighbors} \subseteq \channels$; $p.root \in\{true, false\}$ \\
     \emph{/* $p.\mathit{neighbors}$, as other sets below, are implemented as lists */}
  \end{tabular} \hfill\ 

\smallskip
  
  \textbf{Local Variables:} \hfill\
  
  \begin{tabular}{l}
    $p.d \in \mathds N$; $p.par \in \channels$
  \end{tabular} \hfill\ 

\smallskip
  
  \textbf{Macros:} \hfill\
  
  \begin{tabular}{l}
    $Dist_p = \min \{ q.d + 1, q \in p.\mathit{neighbors} \}$ \\
    $Par_{dist}$ returns the first channel in the list $\{ q \in p.\mathit{neighbors},
    q.d + 1 = p.d \}$
  \end{tabular} \hfill\ 

  \smallskip
  
  \textbf{Action for the root, \ie, for $p$ such that $p.root = true$} \hfill\ 

  \begin{tabular}{ll}
    Action $Root$: & \textbf{if} $p.d \neq 0$ \textbf{then} $p.d := 0$
 \end{tabular} \hfill\ 

  \smallskip
  
  \textbf{Actions for any non-root node, \ie, for $p$ such that $p.root = false$} \hfill\ 
  
  \begin{tabular}{ll}
   Action $CD$:
    & \textbf{if} $p.d \neq Dist_{p}$ \textbf{then} $p.d := Dist_{p}$ \\
     Action $CP$:
    &  \textbf{if} $p.d = Dist_p$ and $p.par.d + 1 \neq p.d$ \textbf{then} $p.par := Par_{dist}$ 
  \end{tabular}  \hfill\ 

   \caption{Algorithm \BFS, code for each node $p$.}
  \label{alg}
\end{algorithm}

In the \emph{atomic-state model}, nodes communicate through locally
shared variables: a node can read its variables and the ones of its
neighbors, but can only write to its own variables. Every node can
access the variables of its neighbors through local channels, denoted
by the set $\channels$ in Algorithm~\ref{alg}.
The network is locally defined at each node $p$ using constant local
inputs.  The fact that the network is rooted is implemented using a
constant Boolean input called $p.root$ which is false for every node
except \Root. The input $p.\mathit{neighbors}$ is the set of channels linking
$p$ to its neighbors.  When it is clear from the context, we do not
distinguish a neighbor from the channels to that neighbor.

The code of Algorithm~\ref{alg} is given as three
locally-mutually-exclusive actions written
as: \textbf{if} \emph{condition} \textbf{then} \emph{statement}. We
say that an action is \emph{enabled} when its condition is true. By
extension, a node is said to be enabled when at least one of its
actions is enabled.
According to the algorithm, the \emph{semantics of the system} defines an
execution as follows.  The system
current \emph{configuration} is given by the current value of all
variables at each node.  If no node is enabled in the current
configuration, then the configuration is said to be \emph{terminal}
and the execution is over.  Otherwise, a \emph{step} is performed:
a \emph{daemon} (an oracle that models the asynchronism of the
system) \emph{activates} a non-empty set of enabled nodes.  Each
activated node then \emph{atomically executes} the statement of its
enabled action, leading the system to a new configuration.

Assumptions can be made about the daemon. Here, we consider the most
general asynchrony assumption, namely the \emph{unfair} daemon,
meaning that it can choose any non-empty subset of the enabled nodes
for execution. In contrast, \emph{fair} daemons would guarantee
additional properties.  For example, a
\emph{strongly} (resp. \emph{weakly}) \emph{fair} daemon ensures that
every node that is enabled infinitely (resp. continuously) often is
eventually chosen for execution by the daemon.

In Algorithm \BFS, each node $p$ maintains two variables. First, it
evaluates in $p.d$ its distance to the root. Then, it maintains
 $p.par$ as a pointer to its \emph{parent} in the tree under
construction: $p.par$ is assigned to a neighbor that is closest to the
root (\nb, \Root.$par$ is meaningless).
Algorithm \BFS is a self-stabilizing BFS spanning tree construction in
the sense that, regardless the initial configuration, it makes the
system converge to a terminal configuration where $par$-variables
describe a BFS spanning tree rooted at \Root.
To that goal, nodes first compute into their $d$-variable their distance
to the root. The root simply forces the value of \Root.$d$ to be 0;
see  Action $Root$. Then, the $d$-variables of other nodes
are gradually corrected: every non-root node $p$ maintains $p.d$ to
be the minimum value of the $d$-variables of its neighbors incremented
by one; see $Dist_{p}$ and Action $CD$.
In parallel, each non-root node $p$ chooses as parent a neighbor $q$
such that $q.d = p.d-1$ when $p.d$ is locally correct \ie, $p.d =
Dist_{p}$) but $p.par$ is not correctly assigned \ie, $p.par.d$ is not
equal to $p.d-1$); see Action $CP$.

%% file: padec.tex
PADEC~\cite{ACD7} is a general framework, written in
Coq \cite{coqart}, to develop mechanically checked proofs of
self-stabilizing algorithms.  It includes the definition of the
atomic-state model and its semantics, tools for the definition of the
algorithms and their properties, lemmas for common proof patterns, and
case studies.  Definitions in PADEC are designed to be as close as
possible to the standard usage of the self-stabilizing community.
Moreover, it is made general enough to encompass many usual hypothesis
(\eg, about topologies or daemons).

In PADEC, the finite network is described using types \Nodes
and \Channels, 
which respectively represent the nodes and the links between nodes.
The distributed algorithm is defined by providing a local
algorithm at each node. This latter is defined using a type
\States 
that represents the local state of a node
\ie, the values of its local variables and a function $\mathit{run}$
that encodes the local algorithm itself and computes a new state
depending on the current state of the node and that of its neighbors.

The model semantics defines a \emph{configuration} as a function
from \Nodes to \States that provides the local state of each node.
The type of a configuration is given by
$\Env \isdef \Nodes \rightarrow \States$.  An \emph{atomic step} of
the distributed algorithm is encoded as a binary relation over
configurations, denoted by $\Step \subseteq \Env \times \Env$, that
checks the conditions given in the informal model; see
Section~\ref{sec:dolev}.  An \emph{execution} $e$
is a finite or infinite stream of configurations, which models a
\emph{maximal} sequence of configurations where any two consecutive
configurations are linked by the $\Step$ relation.  ``Maximal'' means
that $e$ is finite if and only if its last configuration is
terminal. We use the coinductive\footnote{Coinduction allows to define
and reason about potentially infinite objects.}  type $\mathit{Exec}$
to represent an execution stream along with a coinductive predicate
$\mathit{isExec}$
to check the above condition.
Daemons are also defined as predicates over executions (in the case of
the unfair daemon, this predicate is simply equal
to $\mathit{true}$).

Self-stabilization in PADEC is defined according to the usual
practice: the property is formalized as a predicate
$(\mathit{selfStabilization} \;\; \mathit{SPEC})$
where $\mathit{SPEC}$ is a predicate over executions
and models the specification of the algorithm.  An algorithm
is \emph{self-stabilizing w.r.t. the specification}
$\mathit{SPEC}$ if there exists a set of legitimate configurations
that satisfies the following three properties in every
execution $e$:
\begin{itemize}
\item \underline{\emph{Closure}}:
  if $e$ starts in a legitimate configuration then $e$ only contains
  legitimate configurations;
\item \underline{\emph{Convergence}}:
  $e$ eventually reaches a legitimate configuration; and
\item \underline{\emph{Specification}}:
  if $e$ starts in a legitimate configuration then $e$ satisfies the
  intended specification w.r.t. $\mathit{SPEC}$.
\end{itemize}
An algorithm is said to be \emph{silent} when each of its executions
eventually reaches a terminal configuration; in such a case, the set
of legitimate configurations can be chosen as the set of terminal
configurations.  The closure, convergence, and silent properties are
expressed using Linear Time Logic operators provided in the PADEC
library.

\subsection*{The \BFS Algorithm in PADEC}

For the \BFS Algorithm and its specification, we use the formal encoding
provided in \cite{AltisenCD23}; in particular, the algorithm is a
straightforward faithful translation in Coq of
Algorithm \ref{alg}. Notably, an element of \States,
namely a state of a given node, is a tuple
$(d, \mathit{par}, \mathit{root}, \mathit{neighbors})$ representing
the variables of the node as in Algorithm~\ref{alg}.

As the constant variables $\mathit{root}$ and $\mathit{neighbors}$
represent the network, the assumptions that this network is rooted,
bidirected and connected is encoded in a predicate on a configuration
using only those variables. This predicate, in particular uses the set
of edges of the network $\Edges \isdef \{ (p, q) \;|\; p,
q \in \Nodes \;\wedge\; (p \in q.\mathit{neighbors} \;\vee\; q \in
p.\mathit{neighbors}) \}$. Globally in this precidate, the neighbor
links represent a bidirected connected graph and the Boolean
$\mathit{root}$ should be true for a unique node.
We will assume moreover that this predicate holds for any
configuration, even if this is no more mentionned in the sequel.

In \cite{AltisenCD23}, the \BFS Algorithm was proven using PADEC to be
self-stabilizing and silent for the specification of a BFS spanning
tree, \emph{under the assumption of a weakly fair daemon}.  We extend
here this result to the \emph{unfair daemon}.  Note that, since \BFS
is silent, the properties of closure and specification still hold,
henceforth, relaxing the assumption from a weakly fair to an unfair
daemon is trivial.  The only missing property is the convergence.  The
rest of the paper is therefore focusing on proving the convergence of
the \BFS Algorithm under an unfair daemon in PADEC, \ie, providing a
constructive proof under the form of a potential function and its
corresponding order.

%% file: proof.tex
An execution is fully defined by the \Step relation, as the unfair
daemon does not add any other restriction. As a consequence, the
convergence property can be expressed by the fact that the \Step
relation is well-founded:
\begin{equation}\label{ass:well-founded:step}
  \WellFounded~ \Step
\end{equation}
This means that any execution $e \isdef \gamma_0 \xstep{\Step} \gamma_1
\xstep{\Step} ...$ is finite.  $\WellFounded$ comes from the Coq
standard library where it is expressed as $(\WellFounded~ R) \isdef
(\forall x. ~\Acc~ R~ x)$ for a given relation $R$. $\Acc$ is an
inductive predicate from the Coq standard library as well, and $(\Acc~
R~ x)$ means that every sequence of elements starting from $x$ and
linked by $R$ is finite.

In the following, we will prove the assertion
(\ref{ass:well-founded:step}).  To this end, we will consider the
partitioning of the $\Step$ relation as \( \RStep \cup \DStep \cup
\PStep \) denoting respectively \emph{root steps}, \emph{d-steps}
and \emph{par-steps} defined as follows:
\begin{itemize}
\item $\RStep$ holds for any step $\gamma \xstep{\Step} \gamma'$ in which the
  root executes \ie, such that $\gamma.\r.d \neq \gamma'.\r.d$. Note
  that any subset of non-root nodes may also execute during this step.
\item $\DStep$ holds for any step $\gamma \xstep{\Step} \gamma'$ in which the
  root does not execute and at least one non-root node executes a
  $CD$-action \ie, $\gamma.\r.d = \gamma'.\r.d$ and $\exists
  p.~ \gamma.p.d \neq \gamma'.p.d$.  Note that any subset of non-root
  nodes may also execute either Action $CD$ or $CP$.
\item $\PStep$ holds for any step $\gamma \xstep{\Step} \gamma'$ where $d$
  variables are left unchanged \ie, $\forall p.~ \gamma.p.d
  = \gamma'.p.d$. This implies that a node which executes is not the
  root and executes its $CP$-action.
\end{itemize}
In addition, we will use the following general result, (developped as
a tool in PADEC), which gives sufficient conditions ensuring the union
of two relations is well-founded. This tool has first been developped
in PADEC for algorithms with prioritized rules (as Actions $CD$ and
$CP$) and has been enhanced for this proof.
\begin{proposition}\label{prop:well-founded}
  Let $R_1, R_2$ be relations, $x$ an element. Assume that (1)
  $R_2$ is well-founded and (2) there exist a set $B_1$ and relations
  $R_1'$ well-founded, $E_1$ transitive such that
\begin{enumerate}[label=(2.\roman*)]
\item $x \in B_1$ and for all elements $a$, $b$ if $a \in B_1$ and
  $a \xstep{R_1 \cup R_2} b$ then $b \in B_1$,
\item for all elements $a$, $b$ if $a \in B_1$ and $a \xstep{R_1}
  b$ then $a \xstep{R_1'} b$,
\item for all elements $a$, $b$, $c$, if $a \xstep{R'_1} b$ and $b
  \xstep{E_1} c$ then $a \xstep{R'_1} c$,
\item for all elements $a$, $b$ if $a \xstep{R_2} b$ then
  $a \xstep{E_1} b$.
\end{enumerate}
We can conclude that $(\Acc~ (R_1 \cup R_2)~ x)$ holds.
\end{proposition}
\begin{proof}
  Intuitively, the set $B_1$ represents an over-approximation of the
  elements reachable from $x$ through $R_1 \cup R_2$
  (\textit{2.i}). The relation $R'_1$ represents an abstraction of the
  relation $R_1$ when restricted to the set $B_1$ (\textit{2.ii}).
  The relation $E_1$ can be understood as an equality with respect to
  $R_1'$ (\textit{2.iii}) which moreover abstracts the relation $R_2$
  (\textit{2.iv}).  Usually, $R_1'$ and $E_1$ can be derived from a
  potential function on the set of elements and its induced partial
  order and equality.
  
  The result is then directly obtained by considering the relation
  $<_{lex}$ defined on pairs of elements by
  $$(a,b) <_{lex} (c,d) \isdef a \xstep{R'_1} c \mbox{ or } (a
  \xstep{E_1} c \mbox{ and } b \xstep{R_2} d)$$ and the key
  observations that:
\begin{itemize}
\item $<_{lex}$ is a well-founded lexicographic order since $R'_1$ and
  $R_2$ are well-founded;
\item for all elements $a$, $b$, if $a\in B_1$ and $a \xstep{R_1 \cup
  R_2} b$ then $(b, b) <_{lex} (a, a)$. \qed
\end{itemize}
\end{proof}
As a corollary, using the same notations as in
Proposition~\ref{prop:well-founded}, being given two relations $R_1$
and $R_2$, if for every $x$ we can effectively construct the set $B_1$
and the relations $R'_1$, $E_1$ such that all assumptions are met
(\textit{1}, \textit{2.i} to \textit{2.iv}), then we can conclude that
$R_1 \cup R_2$ is well-founded.

We now proceed to the core of the convergence proof and show
progressively that $\PStep$, $\DStep \cup \PStep$ and $\RStep \cup
\DStep \cup \PStep$ are well-founded.

\begin{lemma}\label{lemma:p-steps:wf}
  $\WellFounded~ \PStep$.
\end{lemma}
\begin{proof}
  We use the potential function denoted $\#CP(\gamma)$ which counts for a given
  configuration $\gamma$ the number of nodes for which the guard of
  their $CP$-action is enabled. We say that such a node is $CP$-enabled,
  otherwise it is $CP$-disabled.

  Considering $par$-steps only \ie, the values of $d$-variables are left
  unchanged then either (i) a node is $CP$-disabled and will remain so
  or (ii) it is $CP$-enabled and if it executes, it becomes $CP$-disabled,
  else it remains $CP$-enabled.  Hence, for every two configurations
  $\gamma$ and $\gamma'$ such that $\gamma \xstep{\PStep} \gamma'$
  we have $\#CP(\gamma') < \#CP(\gamma)$, namely $\#CP$ is
  decreasing. As $\#CP$ is obviously lower-bounded by 0, this ensures that
  $\PStep$ is well-founded. \qed
\end{proof}

To proceed on the next phase of the convergence proof, we will use a
result about executions consisting of $d$-steps only.  The next
proposition guarantees the well-foundedness of the relation $\DStep$
through the existence of an effectively constructive potential
function and its ordering:
\begin{quote}
  \begin{proposition}\label{prop:dstep-potential}
    Given a configuration $\gamma_0$, we can effectively construct:
    \begin{enumerate}[label=(\alph*)]
    \item a set of configurations $B(\gamma_0)$ containing $\gamma_0$
      and closed by taking d- or par-steps,
    \item a potential function on configurations from \Env to a
      domain $D(\gamma_0)$, $\dpot{\gamma_0} :
      \Env \rightarrow D(\gamma_0)$, independent on
      \textit{par}-variables, and
    \item a well-founded order $\prec_d$ on $D(\gamma_0)$,
      such that for all $\gamma \in
      B(\gamma_0)$ and $\gamma \xstep{\DStep} \gamma'$, it holds that
      $\dpot{\gamma_0}(\gamma') \prec_d \dpot{\gamma_0}(\gamma)$.
      \end{enumerate}
  \end{proposition}
\end{quote}
The technical details and the complete proof of
Proposition~\ref{prop:dstep-potential} are presented in
Section~\ref{sec:dstep-potential}.

\begin{lemma}\label{lemma:dp-steps:wf}
  $\WellFounded~ (\DStep \cup \PStep)$
\end{lemma}
\begin{proof}
  Assuming Proposition~\ref{prop:dstep-potential} (for now), we must
  prove $(\Acc~(\DStep \cup \PStep)~\gamma_0)$ for an arbitrary
  configuration $\gamma_0$. Therefore, we use
  Proposition~\ref{prop:well-founded} by taking $R_1 \isdef \DStep$,
  $R_2 \isdef \PStep$ and $x \isdef \gamma_0$.  First,
  Lemma~\ref{lemma:p-steps:wf} ensures that $R_2$ is well-founded.
  Using the Proposition~\ref{prop:dstep-potential} above, we define
  the set $B_1$, and the relations $R_1'$ and $E_1$ as
  follows:
  \[ \begin{array}{rcl} B_1 & \isdef &
    B(\gamma_0) \\ \gamma \xstep{R_1'}~\gamma' & \isdef
    & \dpot{\gamma_0}(\gamma') \prec_d \dpot{\gamma_0}(\gamma) \\
    \gamma \xstep{E_1}
    ~\gamma'& \isdef & \dpot{\gamma_0}(\gamma')
    = \dpot{\gamma_0}(\gamma) \end{array} \]
  The guaranties of Proposition~\ref{prop:dstep-potential} allow to
  fulfill the assumptions of Proposition~\ref{prop:well-founded}.
  Indeed, the relation $R_1'$ is well-founded, see
  Proposition~\ref{prop:dstep-potential}(\textit{c}).  The relation
  $E_1$ is transitive by its definition (based on equality of
  potentials).  The assumption (\textit{2.i}), that is, $B_1$ contains
  $x$ and is closed by $R_1$ or $R_2$ steps follows from
  Proposition~\ref{prop:dstep-potential}(\textit{a}).  Assumption
  (\textit{2.ii}), that is, $R_1'$ is an abstraction of $\DStep$ on
  the set $B_1$ holds because of
  Proposition~\ref{prop:dstep-potential}(\textit{c}).  Assumption
  (2.\textit{iii}) holds trivially by the construction of $R_1'$ and
  $E_1$. Last, assumption (2.\textit{iv}) holds as the potential
  function $\dpot{\gamma_0}$ is not depending on the \textit{par}
  variables, that is, remains insensitive to \textit{par}-steps.
  This proves $(\Acc~(\DStep \cup \PStep)~\gamma_0)$ and 
  finally, as the choice of $\gamma_0$ was arbitrary, we
  prove that $\WellFounded~ (\DStep \cup \PStep)$. \qed
\end{proof} 

It remains to take into account the root steps from $\RStep$.
Remind that the root $\r$ can execute at most once in any
execution: either its $d$-variable is 0 from the beginning and
$\r$ never executes; or it is positive and then $\r$ is enabled.
If it executes, the variable is set to 0 and $\r$ is then disabled
forever. The fact that $\RStep$ is well-founded is therefore trivial
to obtain.  

\begin{theorem}
  $\WellFounded~ (\RStep \cup \DStep \cup \PStep)$
\end{theorem}
\begin{proof}
  We use Proposition~\ref{prop:well-founded} by taking
  $R_1 \isdef \RStep$, $R_2 \isdef \DStep \cup \PStep$, and an
  arbitrary configuration $\gamma_0$.  First,
  Lemma~\ref{lemma:dp-steps:wf} ensures that $R_2$ is
  well-founded.  Second, we define the set $B_1 \isdef \Gamma$ and the
  relations $R_1'$ and $E_1$ as follows:
  \[\begin{array}{rcl}
  \gamma \xstep{R'_1} ~\gamma' & \isdef & \gamma'.\r.d < \gamma.\r.d \\
  \gamma \xstep{E_1}  ~\gamma' & \isdef & \gamma'.\r.d = \gamma.\r.d
  \end{array}\]
  Obviously, $R'_1$ is well-founded as observed above, and $E_1$ is
  transitive by definition.  Since $B_1$ contains all configurations,
  the assumption (\textit{2.i}) is trivially satisfied. Also, $R_1
  \subseteq R'_1$ holds by definition of $R_1'$, hence, it implies
  assumption (\textit{2.ii}).  Assumption (\textit{2.iii}) follows from
  definitions as well and assumption (\textit{2.iv}) holds because
  $\gamma.\r.d$ is not changing for any non-root step. \qed
\end{proof}

%% file: dstep-potential.tex
This section is concerned with the proof of
Proposition~\ref{prop:dstep-potential} stated in
Section~\ref{sec:proof}. To this end, we proceed in three steps.
First, we establish a finite over-approximation on the set of the $d$
values that could be possibly reached in an execution involving
$d$-steps only from some initial configuration $\gamma_0$.
Second, we introduce a
partitioning of edges (being either smooth or non-smooth) and prove
some preservation properties along $d$-steps.
Third, we combine the above results to effectively construct a
potential function for $d$-steps and a well-founded order on the
co-domain of this function, ultimately proving
Proposition~\ref{prop:dstep-potential}.

For the sake of readability, we denote $d$-steps $\gamma \xstep{\DStep}
\gamma'$ shortly by $\gamma \dstep \gamma'$.

\subsection{Bounds on Distance Values}
\label{sec:dstep-potential:bounds}

For a configuration $\gamma$, we define the integers $\maxd\gamma
\isdef \max \{\gamma.q.d \mid q \in \nodes\}$, $\mind\gamma \isdef
\min \{\gamma.q.d \mid q \in \nodes\}$, $\sumd\gamma \isdef \sum
\{\gamma.q.d \mid q \in \nodes\}$\footnote{The sum is taken on the
multiset of $d$ values}.  We also define $\dbot{\gamma}$,
$\dtop{\gamma}$ respectively a \emph{bottom} and a \emph{top}
configuration associated to $\gamma$.  These are identical to $\gamma$
except for $d$ values, defined for every node $p$ as follows:
\begin{eqnarray*}
  \dbot{\gamma}.p.d & \isdef & \mind\gamma \\
  \dtop{\gamma}.p.d & \isdef & \left\{ \begin{array}{ll}
    \gamma.p.d & \mbox{if } p = r \\
    \max \{ \gamma.p.d, 1 + \min \{ \dtop{\gamma}.q.d \mid \\
    \hspace{1cm} (p,q)\in \Edges, \dist{p}{\r} = 1 + \dist{q}{\r} \}
    \} &
    \mbox{otherwise,}
  \end{array} \right.
\end{eqnarray*}
where $\dist{q}{\r}$ represents the distance of some node $q$ to the
root \r.  Note that the recursive definition of $\dtop{\gamma}.p.d$ is
well-defined as the recursion is limited to neighbours $q$ of $p$
located at a smaller distance to the root $\r$ than $p$.  Intuitively,
the maximal $d$ value of a non-root node $p$ in some configuration
reachable from $\gamma$ is either its value in $\gamma$ (\ie, it can
be the case when $p$ does not execute) or 1 plus the minimum of the
maximal $d$ values of its neighbors $q$ closer to the root
(see Action $CD$ when $p$ executes).
We define the partial order $\dleq$ on configurations by taking
$$ \gamma_1 \dleq \gamma_2 \isdef \forall q \in \nodes: \gamma_1.q.d
\le \gamma_2.q.d $$
The next lemma states basic properties of the $\dbot{(.)}$,
$\dtop{(.)}$ operators, namely their idempotence and
their monotonicity with respect to $\dleq$.  The proof follows from
definitions and uses induction on nodes according to their
distance to the root.

\begin{lemma} \label{lemma:bounds:basic} ~
  
  \begin{enumerate}[label=(\roman*)]
  \item For all configuration $\gamma$,
  $\dbot{\gamma} \dleq \gamma \dleq \dtop{\gamma}$,
  $\dbot{(\dbot{\gamma})} = \dbot{\gamma}$ and $\dtop{(\dtop{\gamma})}
  = \dtop{\gamma}$.

  \item For all configurations $\gamma_1$ and $\gamma_2$ such that
  $\gamma_1 \dleq \gamma_2$, $\dbot{\gamma_1} \dleq \dbot{\gamma_2}$ and
  $\dtop{\gamma_1} \dleq \dtop{\gamma_2}$.

\end{enumerate}
\end{lemma}

The next lemma relates the bottom and top configurations to $d$-steps.
The proof is done by induction respectively, on the set of nodes
according to their distance to the root (i) and on the length of an
execution sequence from $\gamma_0$ (ii).\footnote{$\gamma_0 \dstepstar
\gamma$ means that $\gamma$ is reachable from $\gamma_0$ using a
finite number of $d$-steps.}
\begin{lemma}\label{lemma:bounds:dstep} ~
  
  \begin{enumerate}[label=(\roman*)]  
  \item For all configurations $\gamma$ and $\gamma'$ such that
  $\gamma \dstep \gamma'$, $\dbot{\gamma} \dleq \dbot{\gamma'}$ and
  $\dtop{\gamma'} \dleq \dtop{\gamma}$.

  \item For all configurations $\gamma_0$ and $\gamma$ such that
  $\gamma_0 \dstepstar \gamma$,
  $\dbot{\gamma_0} \dleq \gamma \dleq \dtop{\gamma_0}$.

  \end{enumerate}
\end{lemma}

\subsection{Smooth and Non-smooth $d$-steps}
\label{sec:dstep-potential:smooth}

We say that an edge $(p,q)\in\Edges$ is \emph{smooth}
(resp. \emph{non-smooth}) in a configuration $\gamma\in\Env$ if the
difference (in absolute value, $abs$) between the $d$-values at its
endpoints $p$, $q$ is at most 1 (resp. at least 2).  Formally,
consider the predicate
$$\csmooth{\gamma}{(p,q)} \isdef (abs(\gamma.p.d - \gamma.q.d) \le 1).$$
We say that a $d$-step $\gamma \dstep \gamma'$ is \emph{smooth} if all the
nodes $p$ changing their values from $\gamma$ to $\gamma'$ are
connected to smooth edges only in $\gamma$, formally:
$$\begin{array}{l}
  \ssmooth{\gamma \dstep \gamma'} \isdef \\
  \hspace{1cm} \forall p \in \Nodes: (\gamma'.p.d \not= \gamma.p.d) \Rightarrow
  (\forall q \in p.neighbors: \csmooth{\gamma}{(p,q)}
\end{array}$$
We define the rank of an edge $(p,q)\in\Edges$ in a configuration
$\gamma\in\Env$ as
$\crank{\gamma}{(p,q)} \isdef \min(\gamma.p.d, \gamma.q.d)$.

\begin{figure}[th]
  \centering
  \scalebox{0.9}{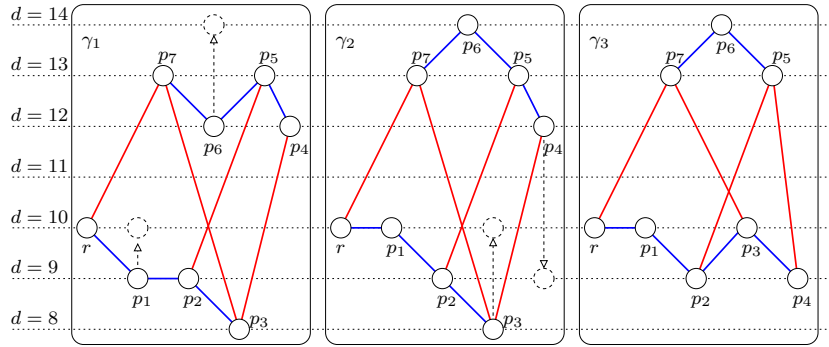}
  \caption{\label{fig:d-steps}Smooth and non-smooth steps}
\end{figure}

For illustration, consider the three configurations $\gamma_1$,
$\gamma_2$, $\gamma_3$ depicted in Fig.~\ref{fig:d-steps}.  We
represented the $d$ values of the nodes by their positioning on the
horizontal lines e.g., $\gamma_1.\r.d = 10$, $\gamma_1.p_1.d = 9$,
$\gamma_2.p_1.d = 10$, etc.  Edges are represented by lines
connecting nodes: smooth (resp. non-smooth) edges are depicted in
blue (resp. red).  Configuration $\gamma_2$ is the successor of
$\gamma_1$ by a smooth step.  That is, only $p_1$ and $p_6$ have
executed and these nodes were connected only to smooth (blue) edges
in $\gamma_1$.  Configuration $\gamma_3$ is the successor of
$\gamma_2$ by a non-smooth step.  That is, $p_3$ and $p_4$ have been
executed along the step, and these nodes were connected to some
non-smoth edges.

The next lemmas provide key properties for understanding the execution
of $d$-steps, depending if they are smooths or not.
Lemma~\ref{lemma:dsteps:smooth} basically states that partitioning
between smooth and non-smooth, as well as the rank of every non-smooth
edge is preserved by smooth steps.  In addition, the total sum of $d$
values is increasing along such a step.  

\begin{lemma}\label{lemma:dsteps:smooth}
  Consider a smooth d-step $\gamma \dstep \gamma'$.  Then,
  \begin{enumerate}[label=(\roman*)]
  \item $\forall e \in \edges: \neg \csmooth{\gamma}{e} \Leftrightarrow
    \neg \csmooth{\gamma'}{e}$,
  \item $\forall e \in \edges: \neg \csmooth{\gamma}{e} \Rightarrow
    (\crank{\gamma}{e} = \crank{\gamma'}{e})$,
  \item $\sumd \gamma' > \sumd \gamma$.
  \end{enumerate}
\end{lemma}
\begin{proof}
  The proof follows immediately from the definition of smooth steps
  and/or edges.  First, the fact that non-smooth edges are preserved
  along with their rank in a smooth $d$-step directly comes from the
  definition of a smooth step: since no node connected to a non-smooth
  edge can execute, non-smooth edges remained unchanged.
  Second, we obtain the increasing of the sum of all $d$-values by
  observing that when a node executes in a smooth $d$-step, its $d$
  value increases by one or two (due to its neighbors which are either
  above by one or at the same level of $d$ value). As a smooth
  $d$-step involves at least one such an executing node, $\sumd$
  necessarily increases (since nodes that do not increase $d$ leave it
  unchanged).
  \qed
\end{proof}

For illustration, consider the smooth step depicted in
Fig.~\ref{fig:d-steps}, \ie, between $\gamma_1$ and $\gamma_2$.  It is
rather trivial that, as long as the nodes executing were connected to
smooth edges only (in blue), their execution has no impact on the
non-smooth edges \ie, they remain non-smooth and preserve their rank.
Yet, the overall sum of the $d$ values increases, here because at
least the values of the two moving nodes has increased (by 1 for $p_1$
and by 2 for $p_6$).

Lemma~\ref{lemma:dsteps:nonsmooth} provides a similar preservation
result for non-smooth steps.  In this case, the key property is that
one can effectively compute a bound $k^*$ such that (i) all non-smooth
edges with rank lower than $k^*$ remain non-smooth and preserve their
rank and (ii) the set of non-smooth edges with rank $k^*$ is strictly
decreasing along the step.  The lemma provides both the explicit
definition of $k^*$ as well as the identification of a non-smooth edge
at level $k^*$ which either becomes smooth or gets a reduced rank
after the step, that is, some edge $(p,q)$ for which the minimum is
achieved in the definition of $k^*$.

\begin{lemma}\label{lemma:dsteps:nonsmooth} Consider a non-smooth d-step $\gamma \dstep \gamma'$.   Let
  $$\begin{array}{l} k^* \isdef \min \{ \crank{\gamma}{(p,q)} \mid (p, q) \in \edges:
    \neg \csmooth{\gamma}{(p,q)}, \\
    \hspace{5cm} \gamma'.p.d \not=\gamma.p.d \mbox{ or } \gamma'.q.d \not= \gamma.q.d \}
    \end{array}$$
  Then,
  \begin{enumerate}[label=(\roman*)]
  \item $\forall e \in \edges: (\crank{\gamma'}{e} \le k^* \wedge \neg \csmooth{\gamma'}{e}) \Rightarrow \\
    \hspace*{3cm} (\crank{\gamma}{e} = \crank{\gamma'}{e} \wedge \neg \csmooth{\gamma}{e})$,
  \item $\forall e \in \edges: (\crank{\gamma}{e} < k^* \wedge \neg \csmooth{\gamma}{e}) \Rightarrow \\
    \hspace*{3cm} (\crank{\gamma'}{e} = \crank{\gamma}{e} \wedge \neg \csmooth{\gamma'}{e})$,
  \item $\exists e \in \edges: (\crank{\gamma}{e} = k^* \wedge \neg \csmooth{\gamma}{e}) \wedge \\
    \hspace*{3cm} (\neg \csmooth{\gamma'}{e} \Rightarrow \crank{\gamma'}{e} > \crank{\gamma}{e}))$.
  \end{enumerate}
\end{lemma}
\begin{proof}
  (i) The proof is done by case splitting, considering which endpoints
  of non-smooth edges $e$ execute. In fact, the only feasible case is
  when none of them executes.  In all other cases, by choosing the
  node which gives a new value to its $d$ variable, we obtain a
  contradiction, either with the minimality of $k^*$ or with the
  non-smoothness of $e$ in $\gamma'$.

  (ii) By definition of $k^*$, no node involved in a non-smooth edge
  can execute if the rank is below $k^*$, hence rank and
  non-smoothness are left unchanged.

  (iii) Note here that, using Coq, to be able to prove "$\exists e \in
  \edges: ...$", we have to effectively contruct such an edge. In our
  case, it is chosen as some of the edges which achieves the minimum
  rank value when computing $k^*$: a non-smooth edge $e^*$ such that
  $\crank{\gamma}{e^*} = k^*$, and one of its end nodes executes
  during the step (it exists and can be computed using the computation
  of the minimum value over a finite set). Now, consider the case
  where $e^*$ remains non-smooth in $\gamma'$. We note $e^*=(p, q)$
  with $\crank{\gamma}{(p, q)} = \gamma.p.d$. We can prove that if $p$
  executes then $\gamma'.p.d > \gamma.p.d$ and that if $q$ executes
  then $\gamma'.q.d = \gamma.p.d + 1$ (see Fig.~\ref{fig:dsteps:proof}
  for an illustration). The result is then easy to conclude.
  \qed
\end{proof}

\begin{figure}[th]
  \begin{center}
    \scalebox{0.9}{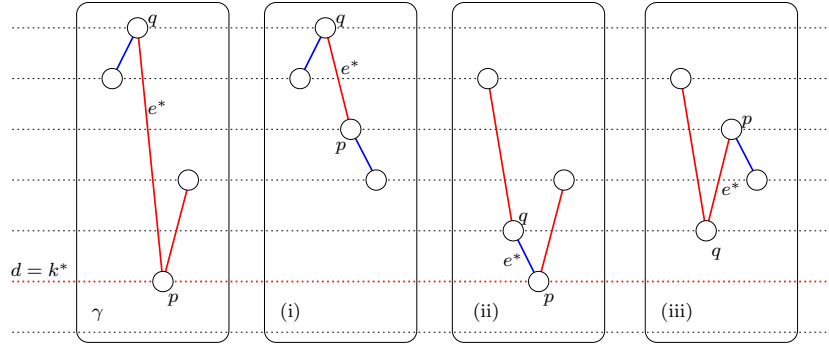}
  \end{center}
  \caption{\label{fig:dsteps:proof}Possible evolutions of a non-smooth edge $e^*=(p,q)$ with
    minimal rank $k^*$: (i) only $p$ executes, (ii) only $q$ executes
    (iii) $p$ and $q$ executes}
\end{figure}

For illustration also, consider the non-smooth step depicted in
Fig.~\ref{fig:d-steps} between $\gamma_2$ and $\gamma_3$.  In this
case, the bound value is $k^* = 8$.  The lemma ensures that the set of
non-smooth edges of rank strictly lower than $8$ are unchanged.  No
such edges actually exist in the configurations $\gamma_2$ or
$\gamma_3$. But, actually, it is not hard to imagine that if such
edges would exist and are not related to $p_3$ and $p_6$, they would
not be impacted by the move.  Also, the lemma guarantees that the set
of edges at level 8 is strictly decreasing.  That is, the set of
non-smooth edges at level 8 is $\{ (p_3,p_7), (p_3,p_4) \}$ in
$\gamma_2$, respectively $\emptyset$ in $\gamma_3$.

Finally, we define $\nsset{\gamma}{k} \isdef \{ e \in \edges \mid \neg
\csmooth{\gamma}{e} \;\wedge\; \crank{\gamma}{e} = k \}$, that is, the
set of non-smooth edges of rang $k$ in $\gamma$.  The next lemma
simply re-formulates the results of Lemma \ref{lemma:dsteps:smooth} in
point (i) and Lemma \ref{lemma:dsteps:nonsmooth} in point (ii) into a
single statement about the sets $\nsset{\gamma}{k}$ to
facilitate their use in the definition of the potential function in
the next subsection.

\begin{lemma}\label{lemma:dsteps}
  Consider a d-step $\gamma \dstep \gamma'$.  Then
  \begin{enumerate}[label=(\roman*)]
  \item if the step $\gamma \dstep \gamma'$ is smooth then
    $\nsset{\gamma}{k} = \nsset{\gamma'}{k}$ for all integer $k$,
  \item if the step $\gamma \dstep \gamma'$ is non-smooth then
    (a) $\nsset{\gamma}{k} = \nsset{\gamma'}{k}$ for all integer $k <
    k^*$ and (b) $\nsset{\gamma'}{k^*} \subsetneq
    \nsset{\gamma}{k^*}$.
  \end{enumerate}
\end{lemma}

\subsection{Potential Function and Proof of
  Proposition~\ref{prop:dstep-potential}} 
\label{sec:dstep-potential:function}

Given a finite interval of integers $K$, and two finite sequences of
$K$-indexed finite sets $\mathcal{X} \isdef (X_k)_{k \in K}$,
$\mathcal{Y} \isdef (Y_k)_{k \in K}$ we write $\mathcal{X} =
\mathcal{Y}$ whenever $X_k = Y_k$ for all $k \in K$, and $\mathcal{X}
\prec_{setlex} \mathcal{Y}$ whenever there exists an integer $k^* \in
K$ such that $X_k = Y_k$ for all $k \in K$, $k<k^*$ and $X_{k^*}
\subsetneq Y_{k^*}$.  Note that $\prec_{setlex}$ is a well-founded
lexicographic order on the set of finite sequences of $K$-indexed
finite sets.

\subsection*{Proof of Proposition~\ref{prop:dstep-potential}}
\begin{proof}
  (a) We define $B(\gamma_0) \isdef \{ \gamma ~|~ \dbot{\gamma_0}
  \dleq \gamma \dleq \dtop{\gamma_0} \}$.  From
  Lemma~\ref{lemma:bounds:basic}(i) we obtain immediately $\gamma_0 \in
  B(\gamma_0)$.  The set $B(\gamma_0)$ is obviously closed by taking
  $par$-steps, as these steps do no change the values of $d$-variables.
  The closure of $B(\gamma_0)$ by $d$-steps can be understood by the
  $\dleq$ inequalities depicted below:

  \begin{center}
    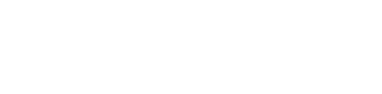
  \end{center}
  
  Knowing $\gamma \in B(\gamma_0)$, that is, $\dbot{\gamma_0} \dleq
  \gamma \dleq \dtop{\gamma_0}$ we obtain the inequalities from the
  top line by using the idempotence and monotonicity of
  $\dbot{(.)}$, $\dtop{(.)}$ with respect to $\dleq$
  (Lemma~\ref{lemma:bounds:basic}). The same lemma ensures the
  inequalities of the bottom line.  Finally, the inequalities across
  the two lines hold because of Lemma~\ref{lemma:bounds:dstep}.  All
  over, they ensure that $\dbot{\gamma_0} \dleq \gamma' \dleq
  \dtop{\gamma_0}$ for any $d$-step $\gamma\dstep\gamma'$.
  
  (b) We define the interval of integers $K_0 \isdef [\mind
    \dbot{\gamma_0}, \maxd \dtop{\gamma_0}]$, that is, the interval of
  possible $d$-values in the configurations reachable from $\gamma_0$.
  We define the domain $D(\gamma_0) \isdef (2^\edges)^{K_0} \times
  [\sumd \dbot{\gamma_0}, \sumd \dtop{\gamma_0}]$. That is,
  $D(\gamma_0)$ consists of pairs $(\mathcal{E},s)$ where $\mathcal{E}
  : K_0 \rightarrow 2^\edges$ is a $K_0$-indexed sequence of sets of
  edges and $s$ is a bounded integer.  In particular, note that
  $D(\gamma_0)$ is finite.  We define the potential function
  $\dpot{\gamma_0} : \Gamma \rightarrow D(\gamma_0)$ by taking
  $\dpot{\gamma_0}(\gamma) \isdef ((\nsset{\gamma}{k})_{k\in K_0},
  \sumd\gamma)$.  Remark that $\dpot{\gamma_0}$ is not dependent on
  \textit{par} variables in $\gamma$.

  (c) We define the relation $\prec_d$ on $D(\gamma_0)$ by taking
  $(\mathcal{E}_1,s_1) \prec_d (\mathcal{E}_2,s_2) \isdef
  \mathcal{E}_1 \prec_{setlex} \mathcal{E}_2 \vee (\mathcal{E}_1 =
  \mathcal{E}_2 \wedge s_2 < s_1)$.  That is, $\prec_d$ is actually a
  strict lexicographic order on pairs $(\mathcal{E},s)$ which combines
  the well-founded order $\prec_{setlex}$ on finite sequences of
  finite sets and a well-founded order $<$ on bounded integers.  It
  remains to prove that $$\forall \gamma,\gamma'\in\Env, ~ \gamma \in
  B(\gamma_0) \mbox{ and } \gamma \dstep \gamma' \Rightarrow
  \dpot{\gamma_0}(\gamma') \prec_d \dpot{\gamma_0}(\gamma)$$ Let
  respectively $(\mathcal{E},s) \isdef \dpot{\gamma_0}(\gamma)$,
  $(\mathcal{E}',s') \isdef \dpot{\gamma_0}(\gamma')$. Note that from
  $\gamma \in B(\gamma_0)$ and the previous point (a) we obtain that
  $\gamma' \in B(\gamma_0)$ as well.  In particular, this ensures the
  ranks of non-smooth edges of $\gamma$, $\gamma'$ are contained in
  $K_0$ and respectively $s$, $s'$ are contained in the interval
  $[\sumd \dbot{\gamma_0}, \sumd\dtop{\gamma_0}]$.
  Lemma~\ref{lemma:dsteps} and Lemma~\ref{lemma:dsteps:smooth}($iii$)
  provide the conditions ensuring that $\dpot{\gamma_0}$ is indeed a
  decreasing potential function with respect to $\prec_d$ as expected.
  For non-smooth steps, we observe the strict inequality $\mathcal{E}'
  \prec_{setlex} \mathcal{E}$. For smooth $d$-steps we observe the
  equality ${\mathcal E} = \mathcal{E}'$ and the strict inequality $s
  < s'$. \qed
\end{proof}

%% file: d-steps.pdf_t
\begin{picture}(0,0)%
\includegraphics{d-steps.pdf}%
\end{picture}%
\setlength{\unitlength}{3108sp}%
\begingroup\makeatletter\ifx\SetFigFont\undefined%
\gdef\SetFigFont#1#2#3#4#5{%
  \reset@font\fontsize{#1}{#2pt}%
  \fontfamily{#3}\fontseries{#4}\fontshape{#5}%
  \selectfont}%
\fi\endgroup%
\begin{picture}(7272,3039)(886,-1558)
\put(901,1334){\makebox(0,0)[lb]{\smash{{\SetFigFont{8}{9.6}{\rmdefault}{\mddefault}{\updefault}{\color[rgb]{0,0,0}$d=14$}%
}}}}
\put(901,884){\makebox(0,0)[lb]{\smash{{\SetFigFont{8}{9.6}{\rmdefault}{\mddefault}{\updefault}{\color[rgb]{0,0,0}$d=13$}%
}}}}
\put(901,434){\makebox(0,0)[lb]{\smash{{\SetFigFont{8}{9.6}{\rmdefault}{\mddefault}{\updefault}{\color[rgb]{0,0,0}$d=12$}%
}}}}
\put(901,-16){\makebox(0,0)[lb]{\smash{{\SetFigFont{8}{9.6}{\rmdefault}{\mddefault}{\updefault}{\color[rgb]{0,0,0}$d=11$}%
}}}}
\put(901,-916){\makebox(0,0)[lb]{\smash{{\SetFigFont{8}{9.6}{\rmdefault}{\mddefault}{\updefault}{\color[rgb]{0,0,0}$d=9$}%
}}}}
\put(901,-1366){\makebox(0,0)[lb]{\smash{{\SetFigFont{8}{9.6}{\rmdefault}{\mddefault}{\updefault}{\color[rgb]{0,0,0}$d=8$}%
}}}}
\put(901,-466){\makebox(0,0)[lb]{\smash{{\SetFigFont{8}{9.6}{\rmdefault}{\mddefault}{\updefault}{\color[rgb]{0,0,0}$d=10$}%
}}}}
\put(1531,-736){\makebox(0,0)[lb]{\smash{{\SetFigFont{8}{9.6}{\rmdefault}{\mddefault}{\updefault}{\color[rgb]{0,0,0}$r$}%
}}}}
\put(1981,-1186){\makebox(0,0)[lb]{\smash{{\SetFigFont{8}{9.6}{\rmdefault}{\mddefault}{\updefault}{\color[rgb]{0,0,0}$p_1$}%
}}}}
\put(2431,-1186){\makebox(0,0)[lb]{\smash{{\SetFigFont{8}{9.6}{\rmdefault}{\mddefault}{\updefault}{\color[rgb]{0,0,0}$p_2$}%
}}}}
\put(3106,974){\makebox(0,0)[lb]{\smash{{\SetFigFont{8}{9.6}{\rmdefault}{\mddefault}{\updefault}{\color[rgb]{0,0,0}$p_5$}%
}}}}
\put(2611,164){\makebox(0,0)[lb]{\smash{{\SetFigFont{8}{9.6}{\rmdefault}{\mddefault}{\updefault}{\color[rgb]{0,0,0}$p_6$}%
}}}}
\put(2206,974){\makebox(0,0)[lb]{\smash{{\SetFigFont{8}{9.6}{\rmdefault}{\mddefault}{\updefault}{\color[rgb]{0,0,0}$p_7$}%
}}}}
\put(3376,164){\makebox(0,0)[lb]{\smash{{\SetFigFont{8}{9.6}{\rmdefault}{\mddefault}{\updefault}{\color[rgb]{0,0,0}$p_4$}%
}}}}
\put(3781,-736){\makebox(0,0)[lb]{\smash{{\SetFigFont{8}{9.6}{\rmdefault}{\mddefault}{\updefault}{\color[rgb]{0,0,0}$r$}%
}}}}
\put(4231,-736){\makebox(0,0)[lb]{\smash{{\SetFigFont{8}{9.6}{\rmdefault}{\mddefault}{\updefault}{\color[rgb]{0,0,0}$p_1$}%
}}}}
\put(4681,-1186){\makebox(0,0)[lb]{\smash{{\SetFigFont{8}{9.6}{\rmdefault}{\mddefault}{\updefault}{\color[rgb]{0,0,0}$p_2$}%
}}}}
\put(5626,164){\makebox(0,0)[lb]{\smash{{\SetFigFont{8}{9.6}{\rmdefault}{\mddefault}{\updefault}{\color[rgb]{0,0,0}$p_4$}%
}}}}
\put(5356,974){\makebox(0,0)[lb]{\smash{{\SetFigFont{8}{9.6}{\rmdefault}{\mddefault}{\updefault}{\color[rgb]{0,0,0}$p_5$}%
}}}}
\put(4906,1064){\makebox(0,0)[lb]{\smash{{\SetFigFont{8}{9.6}{\rmdefault}{\mddefault}{\updefault}{\color[rgb]{0,0,0}$p_6$}%
}}}}
\put(4456,974){\makebox(0,0)[lb]{\smash{{\SetFigFont{8}{9.6}{\rmdefault}{\mddefault}{\updefault}{\color[rgb]{0,0,0}$p_7$}%
}}}}
\put(6031,-736){\makebox(0,0)[lb]{\smash{{\SetFigFont{8}{9.6}{\rmdefault}{\mddefault}{\updefault}{\color[rgb]{0,0,0}$r$}%
}}}}
\put(6481,-736){\makebox(0,0)[lb]{\smash{{\SetFigFont{8}{9.6}{\rmdefault}{\mddefault}{\updefault}{\color[rgb]{0,0,0}$p_1$}%
}}}}
\put(6931,-1186){\makebox(0,0)[lb]{\smash{{\SetFigFont{8}{9.6}{\rmdefault}{\mddefault}{\updefault}{\color[rgb]{0,0,0}$p_2$}%
}}}}
\put(3016,-1366){\makebox(0,0)[lb]{\smash{{\SetFigFont{8}{9.6}{\rmdefault}{\mddefault}{\updefault}{\color[rgb]{0,0,0}$p_3$}%
}}}}
\put(5266,-1366){\makebox(0,0)[lb]{\smash{{\SetFigFont{8}{9.6}{\rmdefault}{\mddefault}{\updefault}{\color[rgb]{0,0,0}$p_3$}%
}}}}
\put(7381,-736){\makebox(0,0)[lb]{\smash{{\SetFigFont{8}{9.6}{\rmdefault}{\mddefault}{\updefault}{\color[rgb]{0,0,0}$p_3$}%
}}}}
\put(7831,-1186){\makebox(0,0)[lb]{\smash{{\SetFigFont{8}{9.6}{\rmdefault}{\mddefault}{\updefault}{\color[rgb]{0,0,0}$p_4$}%
}}}}
\put(7651,974){\makebox(0,0)[lb]{\smash{{\SetFigFont{8}{9.6}{\rmdefault}{\mddefault}{\updefault}{\color[rgb]{0,0,0}$p_5$}%
}}}}
\put(7156,1064){\makebox(0,0)[lb]{\smash{{\SetFigFont{8}{9.6}{\rmdefault}{\mddefault}{\updefault}{\color[rgb]{0,0,0}$p_6$}%
}}}}
\put(6706,974){\makebox(0,0)[lb]{\smash{{\SetFigFont{8}{9.6}{\rmdefault}{\mddefault}{\updefault}{\color[rgb]{0,0,0}$p_7$}%
}}}}
\put(1531,1109){\makebox(0,0)[lb]{\smash{{\SetFigFont{8}{9.6}{\rmdefault}{\mddefault}{\updefault}{\color[rgb]{0,0,0}$\gamma_1$}%
}}}}
\put(3781,1109){\makebox(0,0)[lb]{\smash{{\SetFigFont{8}{9.6}{\rmdefault}{\mddefault}{\updefault}{\color[rgb]{0,0,0}$\gamma_2$}%
}}}}
\put(6031,1109){\makebox(0,0)[lb]{\smash{{\SetFigFont{8}{9.6}{\rmdefault}{\mddefault}{\updefault}{\color[rgb]{0,0,0}$\gamma_3$}%
}}}}
\end{picture}%

%% file: d-steps-proof.pdf_t
\begin{picture}(0,0)%
\includegraphics{d-steps-proof.pdf}%
\end{picture}%
\setlength{\unitlength}{3108sp}%
\begingroup\makeatletter\ifx\SetFigFont\undefined%
\gdef\SetFigFont#1#2#3#4#5{%
  \reset@font\fontsize{#1}{#2pt}%
  \fontfamily{#3}\fontseries{#4}\fontshape{#5}%
  \selectfont}%
\fi\endgroup%
\begin{picture}(7289,3039)(879,-4438)
\put(901,-3841){\makebox(0,0)[lb]{\smash{{\SetFigFont{8}{9.6}{\rmdefault}{\mddefault}{\updefault}{\color[rgb]{0,0,0}$d=k^*$}%
}}}}
\put(2116,-1591){\makebox(0,0)[lb]{\smash{{\SetFigFont{8}{9.6}{\rmdefault}{\mddefault}{\updefault}{\color[rgb]{0,0,0}$q$}%
}}}}
\put(2116,-2401){\makebox(0,0)[lb]{\smash{{\SetFigFont{8}{9.6}{\rmdefault}{\mddefault}{\updefault}{\color[rgb]{0,0,0}$e^*$}%
}}}}
\put(3781,-1591){\makebox(0,0)[lb]{\smash{{\SetFigFont{8}{9.6}{\rmdefault}{\mddefault}{\updefault}{\color[rgb]{0,0,0}$q$}%
}}}}
\put(3826,-2041){\makebox(0,0)[lb]{\smash{{\SetFigFont{8}{9.6}{\rmdefault}{\mddefault}{\updefault}{\color[rgb]{0,0,0}$e^*$}%
}}}}
\put(5401,-3346){\makebox(0,0)[lb]{\smash{{\SetFigFont{8}{9.6}{\rmdefault}{\mddefault}{\updefault}{\color[rgb]{0,0,0}$q$}%
}}}}
\put(2296,-4066){\makebox(0,0)[lb]{\smash{{\SetFigFont{8}{9.6}{\rmdefault}{\mddefault}{\updefault}{\color[rgb]{0,0,0}$p$}%
}}}}
\put(5626,-4066){\makebox(0,0)[lb]{\smash{{\SetFigFont{8}{9.6}{\rmdefault}{\mddefault}{\updefault}{\color[rgb]{0,0,0}$p$}%
}}}}
\put(3781,-2716){\makebox(0,0)[lb]{\smash{{\SetFigFont{8}{9.6}{\rmdefault}{\mddefault}{\updefault}{\color[rgb]{0,0,0}$p$}%
}}}}
\put(7381,-2491){\makebox(0,0)[lb]{\smash{{\SetFigFont{8}{9.6}{\rmdefault}{\mddefault}{\updefault}{\color[rgb]{0,0,0}$p$}%
}}}}
\put(7111,-3661){\makebox(0,0)[lb]{\smash{{\SetFigFont{8}{9.6}{\rmdefault}{\mddefault}{\updefault}{\color[rgb]{0,0,0}$q$}%
}}}}
\put(5266,-3751){\makebox(0,0)[lb]{\smash{{\SetFigFont{8}{9.6}{\rmdefault}{\mddefault}{\updefault}{\color[rgb]{0,0,0}$e^*$}%
}}}}
\put(7201,-3121){\makebox(0,0)[lb]{\smash{{\SetFigFont{8}{9.6}{\rmdefault}{\mddefault}{\updefault}{\color[rgb]{0,0,0}$e^*$}%
}}}}
\put(1621,-4201){\makebox(0,0)[lb]{\smash{{\SetFigFont{8}{9.6}{\rmdefault}{\mddefault}{\updefault}{\color[rgb]{0,0,0}$\gamma$}%
}}}}
\put(3286,-4201){\makebox(0,0)[lb]{\smash{{\SetFigFont{8}{9.6}{\rmdefault}{\mddefault}{\updefault}{\color[rgb]{0,0,0}(i)}%
}}}}
\put(6661,-4201){\makebox(0,0)[lb]{\smash{{\SetFigFont{8}{9.6}{\rmdefault}{\mddefault}{\updefault}{\color[rgb]{0,0,0}(iii)}%
}}}}
\put(4996,-4201){\makebox(0,0)[lb]{\smash{{\SetFigFont{8}{9.6}{\rmdefault}{\mddefault}{\updefault}{\color[rgb]{0,0,0}(ii)}%
}}}}
\end{picture}%

%% file: ineqchain.pdf_t
\begin{picture}(0,0)%
\includegraphics{ineqchain.pdf}%
\end{picture}%
\setlength{\unitlength}{3315sp}%
\begingroup\makeatletter\ifx\SetFigFont\undefined%
\gdef\SetFigFont#1#2#3#4#5{%
  \reset@font\fontsize{#1}{#2pt}%
  \fontfamily{#3}\fontseries{#4}\fontshape{#5}%
  \selectfont}%
\fi\endgroup%
\begin{picture}(3630,886)(2236,-2105)
\put(3601,-1366){\makebox(0,0)[lb]{\smash{{\SetFigFont{8}{9.6}{\rmdefault}{\mddefault}{\updefault}{\color[rgb]{0,0,0}$\dleq$}%
}}}}
\put(4051,-1366){\makebox(0,0)[lb]{\smash{{\SetFigFont{8}{9.6}{\rmdefault}{\mddefault}{\updefault}{\color[rgb]{0,0,0}$\gamma$}%
}}}}
\put(4051,-2041){\makebox(0,0)[lb]{\smash{{\SetFigFont{8}{9.6}{\rmdefault}{\mddefault}{\updefault}{\color[rgb]{0,0,0}$\gamma'$}%
}}}}
\put(3601,-2041){\makebox(0,0)[lb]{\smash{{\SetFigFont{8}{9.6}{\rmdefault}{\mddefault}{\updefault}{\color[rgb]{0,0,0}$\dleq$}%
}}}}
\put(4501,-2041){\makebox(0,0)[lb]{\smash{{\SetFigFont{8}{9.6}{\rmdefault}{\mddefault}{\updefault}{\color[rgb]{0,0,0}$\dleq$}%
}}}}
\put(4951,-2041){\makebox(0,0)[lb]{\smash{{\SetFigFont{8}{9.6}{\rmdefault}{\mddefault}{\updefault}{\color[rgb]{0,0,0}$\dtop{\gamma'}$}%
}}}}
\put(3151,-2041){\makebox(0,0)[lb]{\smash{{\SetFigFont{8}{9.6}{\rmdefault}{\mddefault}{\updefault}{\color[rgb]{0,0,0}$\dbot{\gamma'}$}%
}}}}
\put(3151,-1366){\makebox(0,0)[lb]{\smash{{\SetFigFont{8}{9.6}{\rmdefault}{\mddefault}{\updefault}{\color[rgb]{0,0,0}$\dbot{\gamma}$}%
}}}}
\put(4501,-1366){\makebox(0,0)[lb]{\smash{{\SetFigFont{8}{9.6}{\rmdefault}{\mddefault}{\updefault}{\color[rgb]{0,0,0}$\dleq$}%
}}}}
\put(5401,-1366){\makebox(0,0)[lb]{\smash{{\SetFigFont{8}{9.6}{\rmdefault}{\mddefault}{\updefault}{\color[rgb]{0,0,0}$\dleq$}%
}}}}
\put(4951,-1366){\makebox(0,0)[lb]{\smash{{\SetFigFont{8}{9.6}{\rmdefault}{\mddefault}{\updefault}{\color[rgb]{0,0,0}$\dtop{\gamma}$}%
}}}}
\put(5851,-1366){\makebox(0,0)[lb]{\smash{{\SetFigFont{8}{9.6}{\rmdefault}{\mddefault}{\updefault}{\color[rgb]{0,0,0}$\dtop{\gamma_0}$}%
}}}}
\put(2701,-1366){\makebox(0,0)[lb]{\smash{{\SetFigFont{8}{9.6}{\rmdefault}{\mddefault}{\updefault}{\color[rgb]{0,0,0}$\dleq$}%
}}}}
\put(5086,-1726){\rotatebox{90.0}{\makebox(0,0)[lb]{\smash{{\SetFigFont{8}{9.6}{\rmdefault}{\mddefault}{\updefault}{\color[rgb]{0,0,0}$\dleq$}%
}}}}}
\put(3196,-1591){\rotatebox{270.0}{\makebox(0,0)[lb]{\smash{{\SetFigFont{8}{9.6}{\rmdefault}{\mddefault}{\updefault}{\color[rgb]{0,0,0}$\dleq$}%
}}}}}
\put(4096,-1546){\rotatebox{270.0}{\makebox(0,0)[lb]{\smash{{\SetFigFont{8}{9.6}{\rmdefault}{\mddefault}{\updefault}{\color[rgb]{0,0,0}$\dstep$}%
}}}}}
\put(2251,-1366){\makebox(0,0)[lb]{\smash{{\SetFigFont{8}{9.6}{\rmdefault}{\mddefault}{\updefault}{\color[rgb]{0,0,0}$\dbot{\gamma_0}$}%
}}}}
\end{picture}%

%% file: conclusion.tex
In this paper, we provide the first constructive proof for the
convergence of the Dolev \emph{et al} BFS Spanning Tree algorithm
under the most general execution assumptions (\ie, unfair daemon,
unbounded variables).
The convergence proof has been fully formalized and automatically
checked using the PADEC framework.
Contrarily to many papers about formal certified proofs of distributed
algorithms, we do not formalize and validate an existing
proof. Rather, due to the constructive aspect of the proofs allowed in
Coq, we had to develop a new proof: we have defined a novel potential
function, allowing a finer comprehension of the system executions
towards terminal configurations.
We believe that, even though the algorithm time complexity is
exponential, this potential function will open the door to a tighter
complexity analysis.